\documentclass[onecolumn,showpacs,preprintnumbers,amsmath,amssymb]{revtex4-1}

\usepackage{lipsum}
\usepackage{tabularx}
\usepackage{array}
\usepackage{mathrsfs}
\usepackage{amssymb}
\usepackage{amsmath}
\usepackage{cases}  
\usepackage{graphicx}
\usepackage{subfigure} 
\usepackage{dcolumn}
\usepackage{bm}
\usepackage{verbatim} 
\usepackage[dvipdfm,pdfstartview=FitH,colorlinks,linkcolor=blue,anchorcolor=blue,citecolor=blue]{hyperref}
\usepackage{amsthm}  
\newtheorem{theorem}{Theorem}

\begin{document}

\title{Scalable orbital-angular-momentum sorting without destroying photon states}

\author{Fang-Xiang Wang}
\author{Wei Chen}
\email{weich@ustc.edu.cn}
\author{Zhen-Qiang Yin}
\author{Shuang Wang}
\author{Guang-Can Guo}
\author{Zheng-Fu Han}
\affiliation{Key Laboratory of Quantum Information, University of Science and Technology of China, Hefei 230026, China\\
and Synergetic Innovation Center of Quantum Information $\&$ Quantum Physics, University of Science and Technology of China, Hefei, Anhui 230026, China}


\begin{abstract}

Single photons with orbital angular momentums (OAM) have attracted substantial attention from researchers. A single photon can carry infinite OAM values theoretically. Thus, OAM photon states have been widely used in quantum information and fundamental quantum mechanics. Although there have been many methods for sorting quantum states with different OAM values, the non-destructive and efficient sorter of high-dimensional OAM remains a fundamental challenge. Here, we propose a scalable OAM sorter which can categorize different OAM states simultaneously, meanwhile, preserve both OAM and spin angular momentum (SAM). Fundamental elements of the sorter are composed of symmetric multiport beam splitters (BSs) and Dove prisms with cascading structure, which in principle can be flexibly and effectively combined to sort arbitrarily high dimensional OAM photons. The scalable structures proposed here greatly reduce the number of BSs required for sorting high dimensional OAM states. In view of the non-destructive and extensible features, the sorters can be used as fundamental devices not only for high-dimensional quantum information processing, but also for traditional optics.

\end{abstract}



\maketitle

\section{introduction}
Single photons with orbital angular momentum (OAM) have been prospective high dimensional resources in quantum physics since 1992 \cite{Allen1992}. The Hilbert space dimension of a single photon with OAM can, in principle, be arbitrarily large. This high-dimensional property gives OAM an important role in fundamental studies of quantum mechanics \cite{Leach2010}, high precision optical metrology \cite{D’Ambrosio2013, Lavery2013, Krenn2016}, micromechanics \cite{Galajda2001, Molina-Terriza2007}, quantum cloning \cite{Nagali2009, Nagali2010}, quantum memory \cite{Ding2015, Zhou2015}, quantum computing \cite{Deng2007} and high-dimensional quantum communication \cite{Bourennane2001, Barreiro2008, Dudley2013, Simon2014, Guan2014, Mirhosseini2015}.

In many fields, sorting different OAM states is a fundamental requirement, such as the measurement of high-dimensional Bell states \cite{Dada2011} and quantum communication. In the early days, a fork hologram was a primary tool for sorter single photons with different OAM values \cite{Mair2001, Gibson2004}. The fork hologram adds a spire phase structure to the incident OAM light. If the added azimuthal phase structure is $-l$, then the incident OAM light with quantum number $l$ will be transformed into a zero-order Gaussian light, whereas light with $l'\neq l$ will be transformed into non-zero OAM light. Thus, a single mode fiber after the hologram will collect only the OAM light with quantum number $l$. However, this method can only measure light with one particular OAM state each time \cite{Mair2001}. The improved sorter with a fork hologram can sort several OAM states but with very low efficiency \cite{Gibson2004}. Recently, Berkhout \emph{et al.} introduced conformal transformation into OAM sorter studies \cite{Berkhout2010}. They employed a Cartesian to log-polar transformation to sort different OAM states simultaneously. By combining the Cartesian to log-polar transformation with refractive beam copying, the number of OAM states being sorted simultaneously can reach 25 to 27, and the sorter efficiency is greater than $92\%$ \cite{Mirhosseini2013, Mirhosseini2014}. Although very successful, sorters with conformal transformation destroy the incident photon states. The OAM mode can be retrieved from the Gaussian mode by introducing a desired azimuthal phase, but the external phase added cannot be retrieved as it is usually not prior known in many applications, e.g., in quantum cryptography. This disadvantage limits further applications of OAM states in quantum processing. 

Leach \emph{et al.} demonstrated a non-destructive sorter method utilizing a Mach-Zehnder interferometer (MZI) with two Dove prisms \cite{Leach2002}. This method can in principle sort different OAM states with $100\%$ efficiency. The method has been further developed to demonstrate non-destructive sorters for spin-OAM states \cite{Leach2004, Zhang2014}. Because the MZI is two dimensional, this type of sorter can only sort OAM states with odd $l$ from states with even $l$. Thus, a cascading stage structure with $N-1$ MZIs is needed to sort $N$ OAM states. The cascading stage structure also calls for different changes of azimuthal phases, corresponding to the quantum numbers $l$ of OAM states, and rotating angles in different stages. As OAM states with different values of $l$ evolve with different Gouy phases \cite{Feng2001}, the changing of the azimuthal phase in the different stages decreases the visibility of the MZIs. These disadvantages limit its applications.

In this work, we propose a high-dimensional sorter that sorts different OAM states simultaneously. This sorter is constructed by symmetric $2N$-port beamsplitters (BSs) and Dove prisms. The symmetric $2N$-port BSs are used to realize $N$-dimensional unitary transformations (UTs). The Dove prisms modulate the phases of different OAM states according to the quantum numbers $l$. Since there is no state collapse or active control process, the sorter does not destroy the input OAM states and will not limit the operation rate of the system. We propose three types of cascading-structure OAM sorters, which greatly reduce the resources required and solve the problem of extending the OAM sorter to higher dimensions.

\section{OAM sorter with multiport beamsplitters}
\label{OAMSorter}


\begin{figure}
\centering
\resizebox{8.5cm}{3cm}{\includegraphics{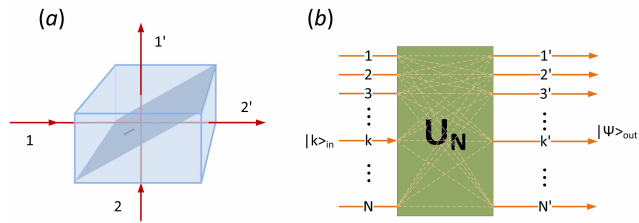}}
\caption{(a) Schematic diagrams of (a) the 4-port BS and (b) the symmetric  $2N$-port BS. The 4-port BS is the element used to construct a $2N$-port BS. A non-OAM photon incident from Port-$k$ of the $2N$-port BS will output from all ports with equal amplitudes.}
\label{fig:BS}
\end{figure}
Before discussing the OAM sorter, we first introduce the $2N$-port BS. Fig. \ref{fig:BS}(a) gives the structure of the elementary 2-dimensional BS (the 4-port BS with two input and two output ports). It can be described by a $2\times 2$ UT:
\begin{equation}
U_2=
\begin{pmatrix}
\sqrt{1-\eta} & -\sqrt{\eta} \\
\sqrt{\eta} & \sqrt{1-\eta}
\end{pmatrix},
\label{eq:U2}
\end{equation}
where $\eta$ and $1-\eta$ are the reflectivity and transmissivity of the interface, respectively. The negative of the last matrix element is introduced by the reflection off the higher index medium. It becomes the most common symmetric 50:50 BS when $\eta=1/2$. Generally, BSs with more input and output ports are called multiport BSs \cite{Walker1986, Walker1987}. We consider only symmetric multiport BSs here. A photon incident from input port $k$ of a symmetric $2N$-port BS will output from all output ports with equal amplitudes (Fig. \ref{fig:BS}(b)). A symmetric $2N$-port BS is described by an $N\times N$ UT. For $N\le 3$, the UT for any symmetric multiport BS is unique based on the conservation of energy. This class of UT is called the canonical multiport \cite{Zukowski1994, Mattle1995}. For $N\ge 4$, there are an infinite number of equivalent classes of UTs for a $2N$-port BS \cite{Bernstein1974, Wootters1986, Torma1996}. Our work is confined to the canonical multiport BS. The $N\times N$ UT of a canonical $2N$-port BS is
\begin{equation}
U_N=\frac{1}{\sqrt{N}}
\begin{pmatrix}
1 & 1 & 1 & \cdots & 1\\
1 & e^{i\frac{2\pi}{N}} & e^{i\frac{4\pi}{N}} & \cdots & e^{i\frac{(N-1)2\pi}{N}}\\
1 & e^{i\frac{4\pi}{N}} & e^{i\frac{8\pi}{N}} & \cdots & e^{i\frac{2(N-1)2\pi}{N}}\\
\vdots & \vdots & \vdots & \ddots & \vdots\\
1 & e^{i\frac{(N-1)2\pi}{N}} & e^{i\frac{(N-1)4\pi}{N}} & \cdots & e^{i\frac{(N-1)^{2}2\pi}{N}}
\end{pmatrix}.
\label{eq:UN}
\end{equation}
For a single photon incident from Port-$k$, the output state will become
\begin{equation}
\begin{aligned}
|\psi\rangle_{out}=U_N |k\rangle_{in}=\frac{1}{\sqrt{N}}\sum_{m=1}^{N}e^{i(k-1)(m-1)2\pi/N}|m\rangle_{out},
\end{aligned}
\label{eq:psi1}
\end{equation}
where $|k\rangle_{in}(|m\rangle_{out})=(0,0,\cdots 0,1_{k(m)},0,\cdots,0)^T$ represents the photon incident (output) from Port-$k(m)$. If $|\psi\rangle_{out}$ goes to a second cascading $2N$-port BS, the final output state will become
\begin{equation}
\begin{aligned}
|\psi\rangle_{out,2}=U_N|\psi\rangle_{out,1}=\begin{cases}
|k\rangle_{out,2} & \text{if $k=1$,}\\
|N-k+2\rangle_{out,2} & \text{if $k\neq 1$.}
\end{cases}
\end{aligned}
\label{eq:psi2}
\end{equation}
The subscript $1(2)$ represents the first (second) $U_N$. It is clear that $|\psi\rangle_{out,2}$ is output from only a single port, as the second $2N$-port BS can be seen as the inverse of the optical path of the first $2N$-port BS. Thus, a single photon incident from a different port of the first $2N$-port BS will output from a different port of the second. This phenomenon is the inspiration of our OAM sorter.

\begin{figure*}
\centering
\resizebox{17cm}{4cm}{\includegraphics{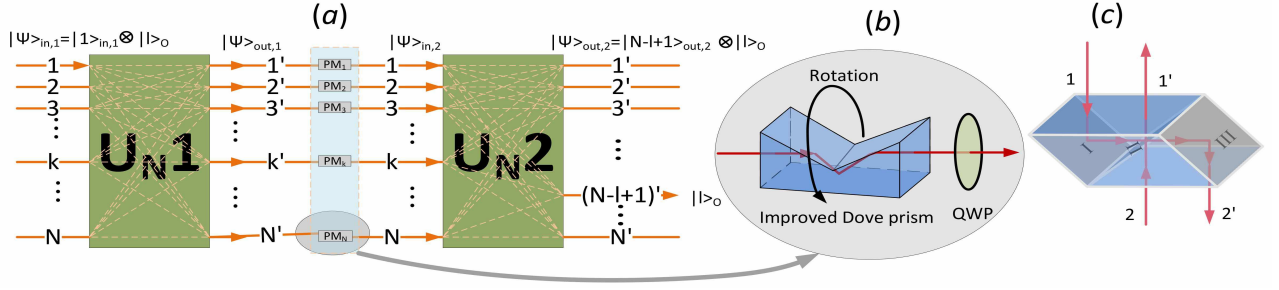}}
\caption{(a) Schematic diagram of the non-destructive OAM sorter. Two cascading $2N$-port BSs with $N$ phase modulators (PMs) are sufficient. (b) The phase modulator (PM) is a SAM preserving module, which consists of an improved Dove prism and a Half-wave plate (HWP). (c) The elemental unitary OAM BS.}
\label{fig:OAMSORTER}
\end{figure*}
Now, we detail the OAM sorter. Eq. \ref{eq:psi1} shows that the first $2N$-port BS acts not only as an equal probability amplitude distributor but also as a phase modulator (PM). The phase introduced depends on the initial incident port. If we add a different phase on a different path, $\phi_m=l(m-1)2\pi/N$ in Path-$m$, between the first and second $2N$-port BS (as shown in Fig. \ref{fig:OAMSORTER}(a)), then $|\psi\rangle_{in,2}=\frac{1}{\sqrt{N}}\sum_{m=1}^{N}e^{i(k+l-1)(m-1)}|m\rangle_{in,2}$, and 
\begin{equation}
\begin{aligned}
|\psi\rangle_{out,2}=\begin{cases} 
|N+2-k-l\rangle_{out,2} & \text{if $k+l<N+2$,}\\
|2N+2-k-l\rangle_{out,2} & \text{if $k+l\geq N+2$,}
\end{cases}
\end{aligned}
\label{eq:psi2OAM}
\end{equation} 
where the integer parameter $l\in [1,N]$. The output port is controlled by parameter $l$. What if the incident photon has OAM and the PMs are Dove prisms (Fig. \ref{fig:OAMSORTER}(b))? When the Dove prism rotates by an angle $\alpha/2$ along the propagation axis, the phase shift of an $l$-order OAM photon becomes $\phi_l=l\alpha$ \cite{Courtial1998a, Courtial1998b}. If the rotation angles of Dove prisms of different paths are in the form $\alpha_m=(m-1)\pi/N$, for an incident $l$-order OAM photon from Port-$k$ of the first $U_N$, the final output port of the second $U_N$ is determined by Eq. \ref{eq:psi2OAM}, provided that the transformation of the $2N$-port BS remains $U_N$. Then, two $2N$-port BSs and $N$ Dove prisms comprise the non-destructive OAM sorter. Additionally, the sorter requires no active modulation of the Dove prism. Thus, there is no limit to the sorter speed. In order to preserve spin angular momentum (SAM), we adopt the improved Dove prism module in Ref. \cite{Leach2004}, which consists of an improved Dove prism and a half-wave plate (HWP), instead of the primary one. For simplicity, the initial OAM states are considered by convention to input from Port-$1$ of the first multiport BS of the OAM sorter if no statement is claimed.


Unfortunately, the classical multiport BSs are constructed from two-dimensional BSs, mirrors and PMs \cite{Zukowski1994, Torma1996, Reck1994}. The BS is not OAM preserving because the reflection on the BS or mirror will invert the $l$-order OAM state $|l\rangle_{O}$ to $|-l\rangle_{O}$, where the subscript $O$ represents the OAM state. There is no interference for two identical OAM states incident from different ports of the BS \cite{Nagali2009}. Thus the multiport BSs can no longer be described by Eq. \ref{eq:UN}. Here we design a simple device that is OAM preserving and can be described by the UT of Eq. \ref{eq:U2}. Fig. \ref{fig:OAMSORTER}(c) gives the schematic diagram of the device, which is a parallelepiped BS. A photon incident from Port-$1$ is totally reflected by the left inner surface I before going to the $\eta$:$(1-\eta)$ BS (interface II), where $\eta$ and $(1-\eta)$ are the reflectivity and transmittivity of interface II, respectively. The transmission part is reflected by the right inner surface III and outputs from Port-$2'$, whereas the reflection part outputs from Port-$1'$ directly. That is, $|l\rangle_{O}\xrightarrow{I}|-l\rangle_{O}\xrightarrow{II}\sqrt{\eta}|l\rangle_{O}+\sqrt{1-\eta}|-l\rangle_{O}\xrightarrow{III}\sqrt{\eta}|l\rangle_{O}+\sqrt{1-\eta}|l\rangle_{O}$. Similarly, if the OAM state $|l\rangle_O$ is incident from Port-$2$, the output state becomes $|l\rangle_{O}\xrightarrow{II}\sqrt{1-\eta}|l\rangle_{O}-\sqrt{\eta}|-l\rangle_{O}\xrightarrow{III}\sqrt{1-\eta}|l\rangle_{O}-\sqrt{\eta}|l\rangle_{O}$. Thus, the parallelepiped device acts as a 2-dimensional UT to the incident OAM state. This device is an elementary OAM BS. Any $2N$-port-OAM BS and UT can thus be realized by the elemental OAM BSs, mirrors and PMs.

\section{scalable structures of OAM sorter}

\begin{figure}
\centering
\resizebox{7.5cm}{11cm}{\includegraphics{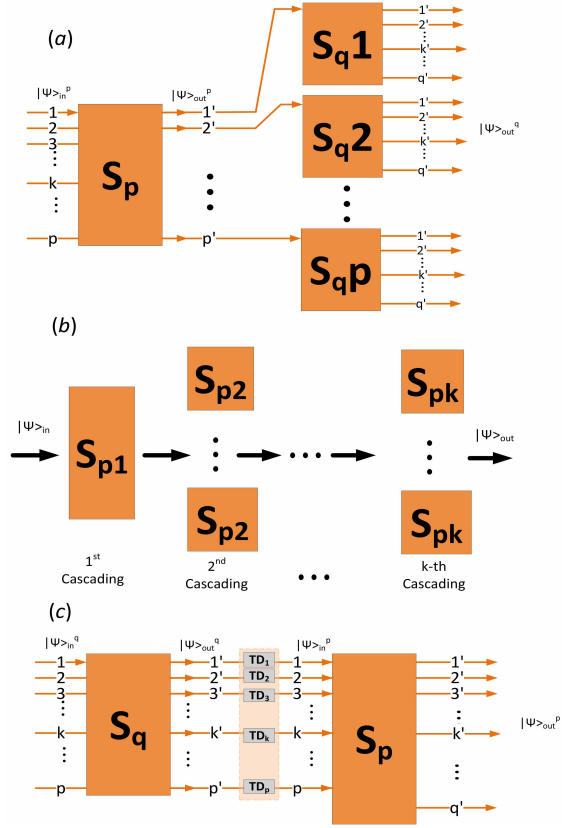}}
\caption{Cascading structure of the OAM sorter for $(p\cdot q)$-dimensional OAM states. (a) Parallel cascading structure with one $S_p$ and $p$ $S_q$. (b) Parallel multi-cascading structure (PMCS). The number of $S_{p_i}$ of the $i$-th cascading is $p_1\cdot p_2\cdots p_{i-1}$, where $i>1$ and $p_1>p_2>\cdots p_k$. (c) Time-delay cascading structure (TDCS) with only one $S_p$ and one $S_q$, where $p>q$. $TD_k$: time delay of Port-$k$.}
\label{fig:CASCADED}
\end{figure}
The number of BSs sufficient for a classical $N$-dimensional UT is ${N\choose 2}$ \cite{Reck1994}, where $N\choose 2$ is the binomial coefficient. However, the number of BSs required increases as $\mathcal{O}(N^2-N)$, causing the $N$-dimensional UT to lack extensibility. This disadvantage limits the applications of multiport BSs for high-dimensional UTs and high-dimensional OAM sorters. This problem can be solved by decomposing the $N$-dimensional OAM sorter into cascading structures, as shown in Fig. \ref{fig:CASCADED}. 

\subsection{Parallel cascading structure}
\label{subPCS}

Fig. \ref{fig:CASCADED}(a) shows the parallel cascading structure (PCS) by which a $p$-dimensional OAM sorter $S_p$ is cascaded by $p$ times a $q$-dimensional sorter $S_q$. As discussed above, the output state from Port-$i'$ of $S_p$ is $|l=pn-i\rangle_O$, and the corresponding output state for $S_q$ is $|l'=qn'-i\rangle_O$, where $n,n'=1,2,3,\cdots$. If Port-$i'$ of $S_p$ is cascaded by Port-1 of $S_q$, the output ports of $S_qi$ for different OAM states are degenerate only if the differences in the quantum values of the OAM states satisfy $l'-l=t[p,q]$, where $t$ is an arbitrary positive integer and $[p,q]$ is the lowest common multiple (LCM) of $p$ and $q$. Thus, if $p$ and $q$ are co-prime, the cascading structure of one $S_p$ and $p$ times $S_q$ can completely sort $(p\cdot q)$-dimensional OAM states. It is easy to prove that $d_{N}^{PCS}(p,q)<d_N$ always holds, where $d_{N}^{PCS}(p,q)$ and $d_N$ are the numbers of BSs necessary for the PCS and the direct $N$-dimensional OAM sorter, respectively. Here we give the optimized condition for the PCS OAM sorter:
\begin{theorem}
For an arbitrary non-prime $N$, if there exists a set of co-prime numbers $\{(p_i,q_i)|p_i>q_i,p_i,q_i\in (1,N)\}$ satisfying $N=p_i\cdot q_i$, then $d_N^{PCS}(p_i,q_i)\ge min\{d_N^{PCS}(p_j,q_j),d_N^{PCS}(p_{j+1},q_{j+1})\}$, where $q_1<q_2<\cdots<q_j<r_0<q_{j+1}<\cdots$, and $r_0$ satisfies $r_0=\sqrt[3]{N+\sqrt{N^2+\frac{1}{27}}}+\sqrt[3]{N-\sqrt{N^2+\frac{1}{27}}}$.
\label{the:PCS}
\end{theorem}
\noindent See Appendix \ref{app:PCS} for the proof of Theorem \ref{the:PCS}.

\subsection{Parallel multi-cascading structure}
\label{subPMCS}

If $p$ or $q$ is not a prime, then $N$ can be decomposed further, $N=\prod_{i=1}^{k}{p_i}^{a_i}$, where $a_i$ is the power of $p_i$, and $p_i$ and $p_j$ are co-prime for any $i\ne j$. If $a_i=1$ for arbitrary $i\in [1,k]$, the PCS can be decomposed into a parallel multi-cascading structure (PMCS) (Fig. \ref{fig:CASCADED}(b)) and the number of BSs necessary, $d_{N}^{PMCS}$, becomes smaller according to Theorem \ref{the:PCS}. Any decomposition resulting in $a_i>1$ is forbidden, as the decomposition reduces the dimensionality $N$ into $N/{p_i}^{a_i-1}$. For example, the optimal PCS for a $20$-dimensional OAM sorter is $(5,4)$. If we decompose the PCS further ($5,2,2$), the dimensionality of the PMCS OAM sorter is reduced into $N'=10$. If $N$ satisfies $N=\prod_{i=1}^{j}p_i, p_i>p_{i+1}$, $p_i$ and $p_j$ are co-prime for any $i\ne j$, then according to Eq. \ref{equ:PCS2}, the optimal $k$-cascading PMCS of a $N$-dimensional OAM sorter is $S_{p_1}\rightarrow p$ times $S_{p_2}\rightarrow (p_1 \cdot p_2)$ times $S_{p_3}\rightarrow\cdots\rightarrow (p_1 \cdot p_2\cdots p_{k-1})$ times $S_{p_k}$.

Generally speaking, the PCS and PMCS given by Theorem \ref{the:PCS} are usually not the optimal PCS and PMCS with the fewest BSs if dimensionality redundancy is allowed. For a particular $N$, there may exist a large factor that cannot be factorized further by Theorem \ref{the:PCS}, while a larger $N$ may avoid this situation. For example, an optimal PMCS 30-dimensional ($5\times 3\times 2$) OAM sorter requires fewer BSs than a 26-dimensional ($13\times 2$) one (Fig. \ref{fig:BSsNumber}). The optimal PCS and 3-cascading PMCS of a $154$-dimensional OAM sorter without dimensionality redundancy are $(22,7)$ and  $(11,7,2)$, respectively. However, if dimensionality redundancy is allowed, a 4-cascading structure ($(7,5,3,2)$) can reduce the number of BSs further. That is to say, a $210$-dimensional OAM sorter requires fewer BSs than a $154$-dimensional one. 

In fact, the PMCS $(5,3,2)$ and $(7,5,3,2)$ are the optimal combinations. That is, no PMCS OAM sorter with $N>30$(210) requires fewer BSs than the $30(210)$-dimensional one with PMCS $(5,3,2)((7,5,3,2))$. 
\begin{theorem}
If $N$ satisfies $N=\prod_{i=1}^{k}p_i, p_i>p_{i+1}$, $p_i$ and $p_{i+1}$ are consecutive primes and $p_k<5$,, the PMCS $N$-dimensional OAM sorter requires fewer BSs than any larger dimensional one. 
\label{the:PMCSOPTI}
\end{theorem}
\noindent See Appendix \ref{app:PMCS}  for the proof of Theorem \ref{the:PMCSOPTI}.

\subsection{Time-delay cascading structure}
\label{subTDCS}

Fig. \ref{fig:CASCADED}(c) shows a time-delay cascading structure (TDCS). In the TDCS, only one $S_p$ and one $S_q$ are needed. The output states of different ports of $S_q$ successively input into the corresponding input ports of $S_p$ according to the TD values. Though requires more time resources, the TDCS greatly reduces the number of BSs. The following theorem gives the optimal TDCS for an $N$-dimensional OAM sorter without dimensionality redundancy:
\begin{theorem}
If $N$ and $\{(p_i,q_i)\}$ satisfy the conditions in Theorem \ref{the:PCS}, then the number of BSs necessary for the TDCS of an $N$-dimensional OAM sorter $d_N^{TDCS}(p_i,q_i)\ge d_N^{TDCS}(p_m,q_m)$, where $p_m$ satisfies $p_m-\sqrt{N}=min\{p_i-\sqrt{N}\}$.
\label{the:TDCS}
\end{theorem}
\noindent 
See Appendix \ref{app:TDCS} for the proof of Theorem \ref{the:TDCS}. 
\begin{figure}
\centering
\resizebox{8cm}{6.5cm}{\includegraphics{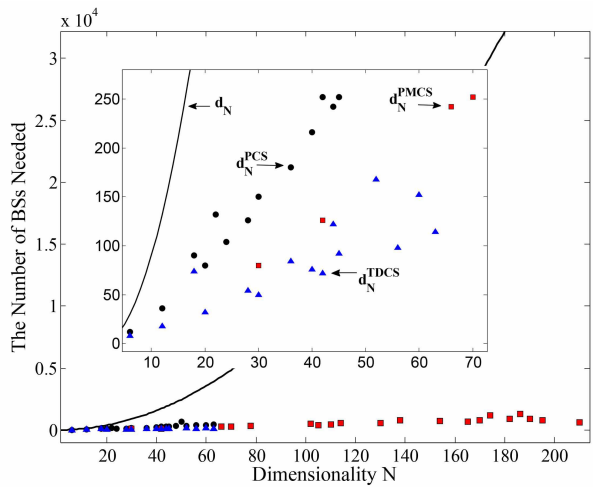}}
\caption{BSs necessary with dimensionality for the OAM sorter with different structures.}
\label{fig:BSsNumber}
\end{figure}

It can be proven that the numbers of BSs necessary for different structures satisfy $d_N=\mathcal{O}(N^2-N), d_N^{PCS}<\mathcal{O}((q+p/q)N)$ and $d_N^{TDCS}<\mathcal{O}((p/q+q/p)N)$. If the factors in the PMCS are all primes, then $d_N^{PMCS}<\mathcal{O}((p_k+4/3)N)$. Fig. \ref{fig:BSsNumber} shows the BSs necessary with dimensionality for OAM sorters of different structures. The cascading structures greatly reduce the number of BSs. For example, when $N=30$, $d_N=870$, while $d_N^{PCS}=150,d_N^{PMCS}=80$ and $d_N^{TDCS}=50$. Even for the much larger dimensionality $N=210$, $d_N^{PMCS}=602$ and $d_N^{TDCS}=392$.

\section{Fidelity of the OAM sorter with imperfect devices}

\begin{figure*}
\resizebox{18cm}{11 cm}{\includegraphics{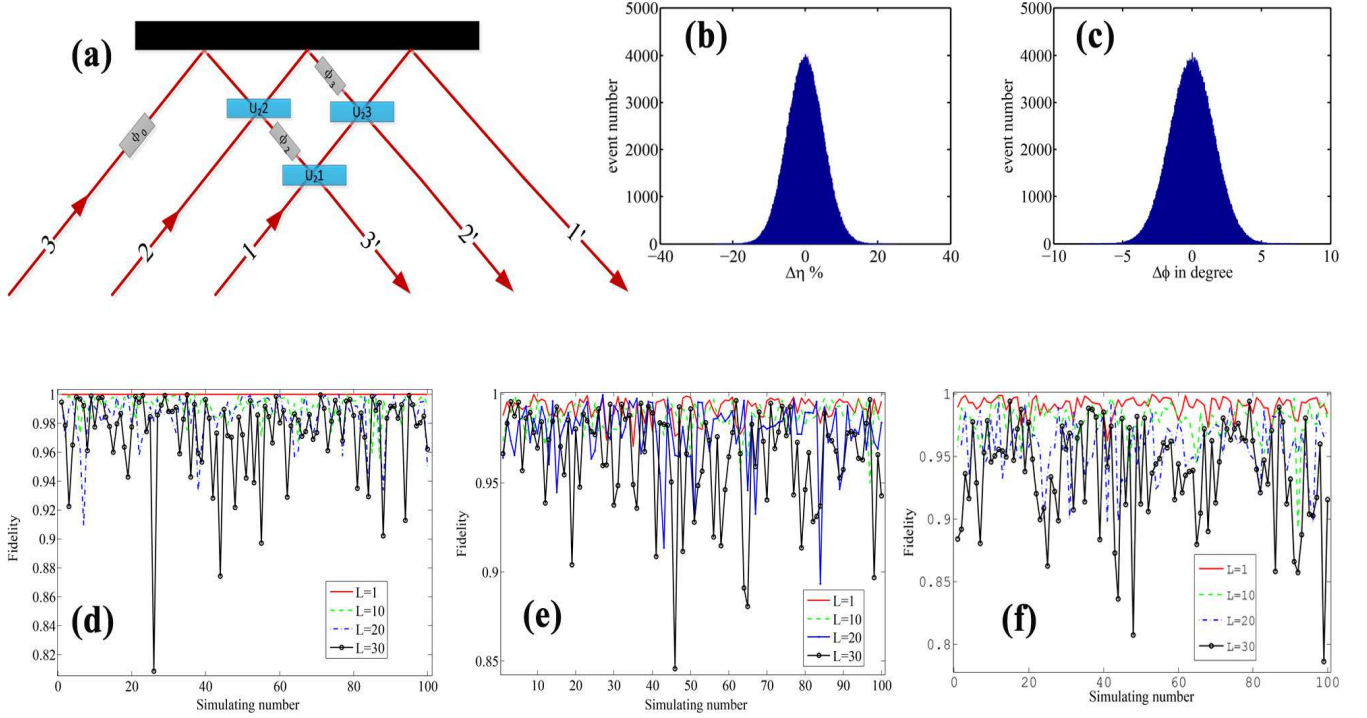}}
\caption{The fidelity analysis of the OAM sorter with imperfect devices. (a) The experimental scheme of a symmetric 3-dimensional BS; (b) The simulated Gaussian distribution of deviations of reflectivities of BSs; (c) The simulated Gaussian distribution of deviations of phase shifters; The random simulation fidelity of (d) the 2-dimensional, (e) the 3-dimensional and (f) the cascading 6-dimensional OAM sorters for different OAM states ($l=1,10,20$ and $30$).}
\label{fig:sortererror}
\end{figure*}

For a practical system, the devices are hardly to meet requirements of the ideal model. Thus, the fidelity is an important index to evaluate the agreement between the practical operation and the desired one. Here, we adopt the definition of the fidelity of a quantum operation in Ref. \cite{Pedersen2007}
\begin{equation}
F=\frac{1}{N(N+1)}[Tr(MM^{\dagger})+|Tr(M)|^2].
\label{eq:fidelity}
\end{equation}
A classical symmetric 3-dimensional BS, expressed by Eq. \ref{eq:UN} is shown in Fig. \ref{fig:sortererror}(a). It consists of 3 BSs ($U_{2}1,U_{2}2$ and $U_{2}3$), 3 phase shifters ($\phi_0,\phi_2$ and $\phi_3$) and a mirror. The parameters are desired as $\eta_1=1/3,\eta_2=\eta_3=1/2,\phi_0=4\pi/3,\phi_2=4\pi/3$ and $\phi_3=5\pi/6$, where $\eta_1,\eta_2$ and $\eta_3$ are the reflectivities of $U_{2}1,U_{2}2$ and $U_{2}3$, respectively. By introducing deviations of the reflectivities and phase shifters, $\Delta\eta_i$ and $\Delta\phi_j$, into the system, the 3-dimensional UT becomes
\begin{equation}
\begin{aligned}
U_3'=&
\begin{pmatrix}
\sqrt{1-\eta_3'} & -\sqrt{\eta_3'} & 0 \\
\sqrt{\eta_3'} & \sqrt{1-\eta_3'} & 0\\
0 & 0 & 1
\end{pmatrix}
\begin{pmatrix}
1 & 0 & 0 \\
0 & e^{i\phi_3'} & 0 \\
0 & 0 & 1
\end{pmatrix}\\
&\begin{pmatrix}
\sqrt{1-\eta_1'} & 0 & -\sqrt{\eta_1'} \\
0 & 1 & 0 \\
\sqrt{\eta_1'} & 0 & \sqrt{1-\eta_1'}
\end{pmatrix}
\begin{pmatrix}
1 & 0 & 0 \\
0 & 1 & 0 \\
0 & 0 & e^{i\phi_2'}
\end{pmatrix}\\
&\begin{pmatrix}
1 & 0 & 0 \\
0 & \sqrt{1-\eta_2'} & -\sqrt{\eta_2'} \\
0 & \sqrt{\eta_2'} & \sqrt{1-\eta_2'} \\
\end{pmatrix}
\begin{pmatrix}
1 & 0 & 0 \\
0 & 1 & 0 \\
0 & 0 & e^{i\phi_0'}
\end{pmatrix},
\end{aligned}
\label{eq:U_3'}
\end{equation}
where $\eta_i'=\eta_i+\Delta\eta_i$, $\phi_j'=\phi_j+\Delta\phi_j$ and $i,j=1,2,3$. The path-depending phase shifters can be expressed as 
\begin{equation}
P_3'=
\begin{pmatrix}
1 & 0 & 0\\
0 & e^{il(2\pi/3+\Delta\phi_{m2})} & 0\\
0 & 0 & e^{il(4\pi/3+\Delta\phi_{m3})}
\end{pmatrix},
\label{eq:P_3'}
\end{equation}
where $\Delta\phi_{mk}$ is the phase deviation of path-$k$ and $k=2,3$. Thus, the practical UT of a 3-dimensional OAM sorter becomes $S_3'=U_3'P_3'U_3'$. If all deviations, $\Delta\eta_i,\Delta\phi_j$ and $\Delta\phi_{mk}$, are equal to 0, $S_3'$ approaches to the desired UT $S_3=U_3P_3U_3$, where $U_3$ and $P_3$ are expressed by Eq. \ref{eq:U_3'} and \ref{eq:P_3'}, respectively, with all deviations being 0. $U_N1$ and $U_N2$ in Fig. \ref{fig:OAMSORTER}(a) have been set as the same one here, as it is easy to realize in experiments by inputting the photon state back into $U_N1$. Hence, the fidelity of a practical 3-dimensional OAM sorter is $F_3=\frac{1}{3\times 4}[Tr(M_3M_3^{\dagger})+|Tr(M_3)|^2]$, where $M_3=S_3^{\dagger}S_3'$.

In order to simulate the fidelity of a practical OAM sorter, we assume that all deviations are independent and follow Gaussian distributions (see Fig. \ref{fig:sortererror}(b) and (c)). By considering the technological level of commercial products, the full width at half maximum (FWHM) of deviations of reflectivities of all BSs are about $11.9\%$ (Fig. \ref{fig:sortererror}(b)). The FWHM of deviations of all phase shifters are about $3.8$ degree (Fig. \ref{fig:sortererror}(c)). The random simulation is repeated for 100 times, and the fidelities for OAM states with $l=1,10,20$ and $30$ are shown in Fig \ref{fig:sortererror}(e).

The random simulation fidelity of a 2-dimensional OAM sorter is shown in Fig. \ref{fig:sortererror}(d). The fidelity of a cascading structure OAM sorter is determined by the number of cascading levels. For a $k$-level cascading structure, the fidelity is $F_N=F_{p_1}F_{p_2}\cdots F_{p_k}$.  Fig. \ref{fig:sortererror}(f) gives the simulation fidelity of a 6-dimensional OAM sorter with 2-level cascading structure ($6=3\times 2$). The average simulation fidelities of Fig. \ref{fig:sortererror}(d), (e) and (f) are shown in TABLE \ref{tab:1}. According to the random simulation results, even the worst fidelity of the cascading structure sorter for a 30-order OAM state is close to 0.8, while the corresponding average fidelity is $<F_6>=0.934$. With state of the art technique, the fidelity can be higher. Thus, the OAM sorter proposed here is realizable. The fidelity of a practical N-dimensional OAM sorter can be analyzed analogously.

On the other hand, the efficiency of practical BSs should be considered for a high-dimensional OAM sorter. For example, the state of art efficiency of a practical BS can be $99\%$, and the total efficiency of a 30-dimensional PMCS OAM sorter ($30=5\times 3\times 2$) becomes $0.99^{5+3+2}=0.904$, which is acceptable for a practical system. However, for a much higher-dimensional sorter, the losses of BSs and any other optical elements make the sorter being less efficiency. Of course, there are some other factors that may decrease the sorting fidelity,such as phase drift with temperature. Thus, temperature control is also necessary for a high-fidelity sorter.
\begin{table}
\centering
\caption{The average random simulation fidelities of the 2-dimensional, 3-dimensional and 6-dimensional OAM sorter.}
\begin{tabular}{ccccc}
\hline \hline
 & $l=1$ & $l=10$ & $l=20$ & $l=30$\\
$<F_2>$ & 1.000 & 0.991 & 0.985 & 0.975\\
$<F_3>$ & 0.990 & 0.986 & 0.977 & 0.958\\
$<F_6>$ & 0.990 & 0.977 & 0.961 & 0.934\\
\hline \hline
\end{tabular}
\label{tab:1}
\end{table}

\section{conclusion}

In conclusion, we have proposed an OAM sorter without destroying the photon states. The sorter consists of $2N$-port BSs and Dove prisms. The cascading structures solve the problem of extending the OAM sorter to higher dimensions. As there is no azimuthal phase change for any OAM state, the Gouy phases of different OAM states do not affect each other. Thus, there is no crosstalk in principle. Since arbitrary discrete finite-dimensional OAM UT as well as the sorters for a single-photon state can be realized by assembling different dimensional multiport BSs \cite{Torma1996,Reck1994}, this structure has a promising prospect for quantum processing. It is worth to be emphasized that the OAM sorter proposed here also preserves SAM, thus the high-dimensional OAM UT is compatible with traditional optics. The multiport BSs used here can be produced by the standard approach shown in Ref. \cite{Reck1994}, which makes the sorters easy to extend. Additionally, by considering the assembly property, the proposed OAM sorters are favorable from the perspective of industrialization.

\noindent \emph{\textbf{Outlook.}} Although traditional integrated optical circuits \cite{Shadbolt2011,Crespi2013,Metcalf2013,Carolan2015} are not compatible with OAM light, special optical fibers for OAM light communication have been reported, and the transmission length record is over 1 kilometer with low crosstalk \cite{Bozinovic2011,Bozinovic2012,Bozinovic2013,Mushref2014}. Furthermore, great progresses has been made on OAM light generation and multiplexing with integrated optical circuit technologies \cite{Grzybowski2012, Guan2014, Pu2015, Ren2016}. With the achievements above, the cascading-structure OAM sorters proposed here possess bright prospects for extending the dimensionality of free-space and on-chip quantum information processing.

\section*{Acknowledgments}
\label{Acknowledgments}

This work has been supported by the National Basic Research Program of China (Grant Nos. 2011CBA00200 and 2011CB921200), the National Natural Science Foundation of China (Grant Nos. 61475148, 61205118, 61675189 and 11304397) and the Strategic Priority Research Program (B) of the Chinese Academy of Sciences (Grant Nos. XDB01030100 and XDB01030300). 

\setcounter{theorem}{0}
\appendix
\numberwithin{equation}{section}
\section{Parallel cascading structure}
\label{app:PCS}
The number of BSs sufficient for a classical $N$-dimensional UT is ${N\choose 2}$ \cite{Reck1994}, where $N\choose 2$ is the binomial coefficient. Thus, the number of BSs required for a direct $N$-dimensional OAM sorter is $d_N=N(N-1)$, increasing as $\mathcal{O}(N^2-N)$. The number of BSs necessary for a PCS OAM sorter is
\begin{equation}
\begin{aligned}
d_N^{PCS}(p,q)&=2({p\choose 2}+p{q\choose 2})=\frac{N}{q}(p-1)+N(q-1)\\
&=(q+\frac{p}{q}-1-\frac{1}{q})N<(q+p/q)N.
\end{aligned}
\label{equ:PCS}
\end{equation}
where $N=pq$, and $p$, $q$ are co-prime. And
\begin{equation}
d_N-d_N^{PCS}(p,q)=p(p-1)(q+1)(q-1)>0,
\label{equ:PCS1}
\end{equation}
as $p,q\geq 2$. Thus, the number of BSs sufficient for a PCS $N$-dimensional OAM sorter is always smaller than a direct $N$-dimensional sorter. There are two combinations of PCS for every $(p,q)$. The combinations are: (a) one $S_p$ and $p$ times $S_q$; (b) one $S_q$ and $q$ times $S_p$. Let $p>q$, then
\begin{equation}
d_N^{PCS}(p,q)-d_N^{PCS}(q,p)=(p-1)(q-1)(q-p)<0
\label{equ:PCS2}
\end{equation}
always holds. Thus, the combination $(p,q)$ requires fewer BSs than that of $(q,p)$.

\begin{theorem}
For an arbitrary non-prime $N$, if there exists a set of co-prime numbers $\{(p_i,q_i)|p_i>q_i,p_i,q_i\in (1,N)\}$ satisfying $N=p_i\cdot q_i$, then $d_N^{PCS}(p_i,q_i)\ge min\{d_N^{PCS}(p_j,q_j),d_N^{PCS}(p_{j+1},q_{j+1})\}$, where $q_1<q_2<\cdots<q_j<r_0<q_{j+1}<\cdots$, and $r_0$ satisfies $r_0=\sqrt[3]{N+\sqrt{N^2+\frac{1}{27}}}+\sqrt[3]{N-\sqrt{N^2+\frac{1}{27}}}$.
\label{appthe:PCS}
\end{theorem}

\begin{proof}
Let $f(N,p,q)=d_N-d_N^{PCS}(p,q)$, then,$f=N\frac{-q^3+Nq^2+q-N}{q^2}$. The derivative of $f$ is $f^{'}=N(-1-\frac{1}{q^2}+\frac{2N}{q^3})$. Let $f^{'}=0$, then
\begin{equation}
-1-\frac{1}{q^2}+\frac{2N}{q^3}=0\Longrightarrow q^3+q-2N=0
\label{equ:f'}
\end{equation}
Equation \ref{equ:f'} is in the form $x^3+ax+b=0$, where $a=1$ and $b=-2N$. The only one real root for this cubic equation is 
\begin{equation}
\begin{aligned}
r_0&=\sqrt[3]{-\frac{b}{2}+\sqrt{(\frac{b}{2})^2+(\frac{a}{3})^3}}+\sqrt[3]{-\frac{b}{2}-\sqrt{(\frac{b}{2})^2+(\frac{a}{3})^3}}\\
&=\sqrt[3]{N+\sqrt{N^2+\frac{1}{27}}}+\sqrt[3]{N-\sqrt{N^2+\frac{1}{27}}}.
\end{aligned}
\label{equ:r0}
\end{equation}
By taking two derivatives of $f$, we obtain $f^{''}=\frac{2N}{q^3}(-\frac{3N}{q}+1)$. Because $q<N$, $f^{''}<0$ always holds. Hence, the root obtained by Eq. \ref{equ:r0} is the maximum point of $f$. In other words, the smallest number of BSs sufficient for a PCS $N$-dimensional OAM sorter is obtained at the point $q=r_0$. However, $r_0$ is not an integer and function $f$ is not symmetric with respect to $r_0$. Thus, the minimum point should be $q_j$ or $q_{j+1}$, where $q_j$ and $q_{j+1}$ satisfy $q_1<q_2<\cdots<q_j<r_0<q_{j+1}<\cdots$.

This completes the proof of Theorem \ref{appthe:PCS}.

\end{proof}

\section{Parallel multi-cascading structure}
\label{app:PMCS}

\begin{theorem}
If $N$ satisfies $N=\prod_{i=1}^{k}p_i, p_i>p_{i+1}$, $p_i$ and $p_{i+1}$ are consecutive primes and $p_k<5$,, the PMCS $N$-dimensional OAM sorter requires fewer BSs than any larger dimensional one. 
\label{appthe:PMCSOPTI}
\end{theorem}

\begin{proof}
We assume there exists a PCS OAM sorter requires fewer BSs than one sub-sorter, $S_{p_j}$, of the PMCS $N$-dimensional sorter. That is, there exists a combination $(a,b)$ satisfying $d_{N^{'}=ab}^{PCS}<d_{N=p_j}$, where $p_j>a>b$. $a$ and $b$ should satisfy $a,b<p_k$, otherwise, the dimensionality of the sorter will be reduced. Because $p_k<5$, then $a\geq 3$ and $p_k\leq 3$, which is contradictory with $a,b<p_k$. Hence, no PCS can substitute for $S_{p_j}$ without reducing the dimensionality of the OAM sorter.

This completes the proof of Theorem \ref{appthe:PMCSOPTI}.

\end{proof}

For a PMCS OAM sorter with $N=\prod_{i=1}^{k}p_i$, the number of BSs necessary is 
\begin{equation}
\begin{aligned}
& d_N^{PMCS}=2({p_1\choose 2}+p_1{p_2\choose 2}+{p_1}{p_2}{p_3\choose 2}+\cdots\\
&\qquad\qquad\quad+{p_1}{p_2}\cdots {p_{k-1}}{p_k\choose 2})\\
&=(\frac{p_1-1}{\prod_{i=2}^{k}p_i}+\frac{p_2-1}{\prod_{i=3}^{k}p_i}+\cdots+\frac{p_{k-1}-1}{p_k}+(p_k-1))N,
\end{aligned}
\label{equ:dnpmcs}
\end{equation}
where $k\geq 3$. According to the list of known prime numbers \cite{Primlist,Primegap}, any consecutive primes $p_i$ and $p_{i+1}$ satisfy $\frac{p_{i+1}}{p_i}\leq \frac{5}{3}$, the equality holds if and only if $p_i=3$ and $p_{i+1}=5$. Hence, Eq. \ref{equ:dnpmcs} becomes

\begin{equation}
\begin{aligned}
d_N^{PMCS}&<[(\frac{1}{\prod_{i=3}^{k}p_i}+\frac{1}{\prod_{i=4}^{k}p_i}+\cdots+\frac{1}{p_k}+1)\times\frac{5}{3}\\
&\quad-(\frac{1}{\prod_{i=2}^{k}p_i}+\frac{1}{\prod_{i=3}^{k}p_i}+\cdots+\frac{1}{p_k}+1)+p_k]N\\
&=[(\frac{1}{\prod_{i=3}^{k}p_i}+\frac{1}{\prod_{i=4}^{k}p_i}+\cdots+\frac{1}{p_k}+1)\times\frac{2}{3}\\
&\quad-\frac{1}{\prod_{i=2}^{k}p_i}+p_k]N\\
&<[(2^{-(k-2)}+2^{-(k-3)}+\cdots+2^{-1}+1)\times\frac{2}{3}\\
&\quad+p_k]N\\
&<(\frac{4}{3}+p_k)N.
\end{aligned}
\label{dnpmcsnum}
\end{equation}

\section{Time-delay cascading structure}
\label{app:TDCS}

In the Time-delay cascading structure (TDCS), only one $S_p$ and one $S_q$ are needed. The number of BSs sufficient for a TDCS $N$-dimensional OAM sorter is
\begin{equation}
\begin{aligned}
d_N^{TDCS}(p,q)&=2({p\choose 2}+{q\choose 2})=p(p-1)+q(q-1)\\
&=\frac{N}{q}(p-1)+\frac{N}{p}(q-1)<(\frac{p}{q}+\frac{q}{p})N.
\end{aligned}
\label{equ:TDCS}
\end{equation}

\begin{theorem}
If $N$ and $\{(p_i,q_i)\}$ satisfy the conditions in Theorem \ref{the:PCS}, then the number of BSs necessary for the TDCS of an $N$-dimensional OAM sorter $d_N^{TDCS}(p_i,q_i)\ge d_N^{TDCS}(p_m,q_m)$, where $p_m$ satisfies $p_m-\sqrt{N}=min\{p_i-\sqrt{N}\}$.
\label{appthe:TDCS}
\end{theorem}

\begin{proof}
Let $g(p,q)=d_N^{TDCS}(p,q)$, then
\begin{equation}
g=p(p-1)+q(q-1)=p^2-p+\frac{N^2}{p^2}-\frac{N}{p}
\label{equ:gpq}
\end{equation}
The derivative of $g$ is $g^{'}=2p-1-\frac{2N^2}{p^3}+\frac{N}{p^2}$. Let $g^{'}=0$, then
\begin{equation}
2p^4-p^3+Np-2N^2=0.
\label{equ:g1}
\end{equation}
The only two real roots of Eq. \ref{equ:g1} are $r_1=\sqrt{N}$ and $r_2=-\sqrt{N}$. $r_2=-\sqrt{N}$ is abandoned as $p$ is a positive integer. Hence, $p=\sqrt{N}$. Taking two derivatives of $g$, we obtain
\begin{equation}
\begin{aligned}
g^{''}&=2+\frac{6N^2}{p^4}-\frac{2N}{p^3}=\frac{2}{p^4}(p^4-Np+3N^2)\\
&>\frac{2}{p^4}(4N^2-Np)>\frac{6N^2}{p^4}>0.
\end{aligned}
\label{g2}
\end{equation}
Thus, $p=\sqrt{N}$ is the minimum point of function $g$. That is, $min\{d_N^{TDCS}(p,q)\}=d_N^{TDCS}(\sqrt{N},\sqrt{N})$ Considering that $p$ is an integer number, and $p$ and $q$ are co-prime, $d_N^{TDCS}(p_i,q_i)\ge d_N^{TDCS}(p_m,q_m)$, where $p_m$ satisfies $p_m-\sqrt{N}=min\{p_i-\sqrt{N}\}$.
This completes the proof of Theorem \ref{appthe:TDCS}.
\end{proof}

\end{document}